\documentclass[11pt]{article}

\usepackage{etex, authblk}
\usepackage{algorithmic, algorithm}
\usepackage{amsmath,amsfonts,amssymb,tikz}
\usepackage[margin=1in]{geometry}
\usepackage{amsthm} 
\usepackage{csquotes}
\usepackage[all]{xy}
\usepackage[bottom]{footmisc}
\usepackage{enumitem}
\usepackage{prettyref} 
\usepackage{hyperref}
\usepackage{bbm, bm}
\usepackage{footnote}

\newtheorem{theorem}{Theorem}[section]

\newtheorem{lemma}[theorem]{Lemma}
\newtheorem{corollary}[theorem]{Corollary}

\newtheorem{claim}[theorem]{Claim}
\newtheorem{proposition}[theorem]{Proposition}
\newtheorem{definition}[theorem]{Definition}
\newtheorem{remark}[theorem]{Remark}

\newtheorem{fact}[theorem]{Fact}

\makeatletter
\newtheorem*{rep@theorem}{\rep@title}
\newcommand{\newreptheorem}[2]{%
\newenvironment{rep#1}[1]{%
\def\rep@title{#2 \ref{##1}}%
\begin{rep@theorem}}%
{\end{rep@theorem}}}
\makeatother

\newreptheorem{theorem}{Theorem}
\newreptheorem{lemma}{Lemma}
\newreptheorem{definition}{Definition}

\newenvironment{proofsketch}{\noindent{\bf Proof Sketch.}}%
{\hspace*{\fill}$\Box$\par}
\newenvironment{proofof}[1]{\smallskip\noindent{\bf Proof of #1.}}%
{\hspace*{\fill}$\Box$\par}

\newcommand{\pref}{\prettyref}
\newrefformat{lem}{Lemma \ref{#1}}
\newrefformat{cl}{Claim \ref{#1}}
\newrefformat{prop}{Proposition \ref{#1}}
\newrefformat{prob}{Problem \ref{#1}}
\newrefformat{thm}{Theorem \ref{#1}}
\newrefformat{cor}{Corollary \ref{#1}}
\newrefformat{cha}{Chapter \ref{#1}}
\newrefformat{sec}{Section \ref{#1}}
\newrefformat{tab}{Table \ref{#1}}
\newrefformat{fig}{Figure \ref{#1}}
\newrefformat{conj}{Conjecture \ref{#1}}
\newrefformat{prot}{Protocol \ref{#1}}
\newrefformat{fact}{Fact \ref{#1}}
\newrefformat{def}{Definition \ref{#1}}
\newrefformat{rem}{Remark \ref{#1}}

\newcommand{\E}{{\mathbb{E}}}
\newcommand{\eps}{\varepsilon}

\newcommand{\cD}{\mathcal{D}}

\newcommand{\cP}{\mathcal{P}}

\newcommand{\cS}{\mathcal{S}}
\newcommand{\cQ}{\mathcal{Q}}
\newcommand{\cU}{\mathcal{U}}
\newcommand{\cT}{\mathcal{T}}

\newcommand{\Patrascu}{P\u{a}tra\c{s}cu}

\newcommand{\ip}[1]{\left\langle #1 \right\rangle}

\newcommand{\Field}{\mathbb{F}}

\newcommand{\abs}[1]{\left| #1 \right|}
\newcommand{\norm}[1]{\| #1 \|}

\newcommand{\KL}[2]{\mathsf{D}_{KL} ( {#1} || {#2} )}

\usepackage{xspace}

\newcommand{\poly}{{\operatorname{poly}\xspace}}

\newcommand{\DISJ}{\mathsf{DISJ}}

\newcommand{\EC}[1]{}

\DeclareMathOperator*{\argmax}{arg\,max}

\floatstyle{boxed}
\newfloat{Protocol}{H}{pro}

\floatstyle{boxed}
\newfloat{Problem}{H}{pro}

\title{Lower Bounds for Linear Operators}
\author{Young Kun Ko}
\affil[]{Department of Computer Science and Engineering, Pennsylvania State University}
\affil[]{Email: ykko@psu.edu}
\date{\today}

\begin{document}

\maketitle

\begin{abstract}



    We consider a static data structure problem of computing a linear operator under cell-probe model. Given a linear operator $M \in \mathbb{F}_2^{m \times n}$, the goal is to pre-process a vector $X \in \mathbb{F}_2^n$ into a data structure of size $s$ to answer any query $\ip{M_i , X}$ in time $t$. We prove that for a random operator $M$, any such data structure requires: 
    $$
        t \geq \Omega ( \min \{ \log (m/s) , n / \log s \} ).
    $$
    This result overcomes the well-known logarithmic barrier in static data structures \cite{miltersen_data_1998, siegel_universal_2004,  patrascu_logarithmic_2006, panigrahy_geometric_2008, patrascu_unifying_2011, dvir_static_2019} by using a random linear operator. Furthermore, it provides the first significant progress toward confirming a decades-old folklore conjecture: that non-linear pre-processing does not substantially help in computing most linear operators.

    A straightforward modification of our proof also yields a wire lower bound of $\Omega(n \cdot \log^{1/d}(n))$ for depth-$d$ circuits with arbitrary gates that compute a specific linear operator $M \in \mathbb{F}_2^{O(n) \times n}$, even against some small constant advantage over random guessing. This bound holds even for circuits with only a small constant advantage over random guessing, improving upon longstanding results \cite{raz_lower_2003, cherukhin_lower_2008, hirsch_lower_2008, gal_tight_2013} for a random operator.

    Finally, our work partially resolves the communication form of the Multiphase Conjecture~\cite{patrascu_mihai_towards_2010} and makes progress on Jukna-Schnitger's Conjecture \cite{jukna_min-rank_2011, jukna_boolean_2012}. We address the former by considering the Inner Product (mod 2) problem (instead of Set Disjointness) when the number of queries $m$ is super-polynomial (e.g., $2^{n^{1/3}}$), and the total update time is $m^{0.99}$. Our result for the latter also applies to cases with super-polynomial $m$. 
    
\end{abstract}

\newpage

\section{Introduction}

A fundamental, long-standing question in the theory of computation is:
\begin{quote}
    Can non-linear computation provide an advantage in computing a linear operator?
\end{quote}

The intuition that non-linear computation offers no significant help has persisted for over 70 years since the seminal work in circuit design by Shannon~\cite{shannon_synthesis_1949}. For instance, Lupanov~\cite{lupanov_rectifier_1956} and later Valiant~\cite{goos_graph-theoretic_1977, valiant_why_1992} drew on this idea to study the complexity of linear operators. This belief underpins major open problems like matrix rigidity--the search for a linear operator that cannot be computed by small, shallow circuits with linear gates (See the surveys \cite{lokam_complexity_2009, ramya_recent_2020} and references therein). Despite its widespread acceptance, this intuition has remained a folklore belief. Only recently did Jukna and Schnitger
\cite{jukna_min-rank_2011, jukna_boolean_2012} formalize it as an informal conjecture: {\bf non-linear gates do not help in computing a linear operator}.

This conjecture extends naturally to the {\bf cell-probe model} \cite{yao_should_1981}, the most powerful model of computation for data structures. In this model, we only count memory accesses (probes), while all computation is free. An input is pre-processed into a data structure of $s$ cells (here, we assume 1-bit cells, as we work over $\Field_2$). To answer a query, an algorithm can probe up to $t$ cells and perform arbitrary computation on their contents. Because of its strength, a lower bound in this model applies to any reasonable data structure.

This raises an analogous question for data structures: for a cell-probe data structure designed to compute $\ip{M_i, X}$ for a linear operator $M$, can non-linear pre-processing of the input $X$ reduce the query time? This can be seen as a cell-probe analogue of the Jukna-Schnitger conjecture (See Section 2 of \cite{dvir_static_2019} for detailed discussion).

Assuming this intuition is true has led to fruitful research on restricted models where pre-processing is linear or from related classes (See \cite{fredman_lower_1981, chazelle_lower_1990, larsen_range_2014, loebl_simplex_2017, dvir_static_2019, afshani_lower_2019, golovnev_polynomial_2022} and references therein). These models are also deeply connected to matrix rigidity \cite{dvir_static_2019, natarajan_ramamoorthy_equivalence_2020}. However, challenging this intuition, there are cases where highly efficient non-linear data structures exist for linear problems with no known linear equivalents \cite{kedlaya_fast_2011}.

Our primary goal is to rigorously establish that for {\bf most} linear operators, {\bf non-linear pre-processing provides no significant advantage}. This, in turn, implies an analogous result for circuit complexity.

\paragraph{Previous Results}

Progress on this front has been stalled by technical barriers in the cell-probe model. The primary tool for static cell-probe lower bounds is essentially a counting argument. A seminal result by Miltersen~\cite{miltersen_bit_1993} shows that for almost all non-linear problems, if the space $s$ is slightly smaller than the number of possible outputs $m$ (e.g., $s=m^{0.99}$), then the query time must be large ($t \geq \Omega ( n^{0.99} )$). If we {\em assume} linear pre-processing is optimal for linear problems, a similar counting argument yields strong lower bounds.

However, this argument breaks down completely when arbitrary non-linear pre-processing is allowed. The number of possible functions to compute even a single pre-processed bit is astronomical ($2^{2^n}$), rendering a simple counting argument over all possible data structures useless.

Prior to our work, the best lower bound for an arbitrary linear problem was derived from techniques used for explicit data structure problems \cite{miltersen_data_1998, siegel_universal_2004, patrascu_logarithmic_2006, panigrahy_geometric_2008, patrascu_unifying_2011}. These techniques hit a ceiling known as the logarithmic barrier, yielding bounds of the form:
\begin{equation} \label{eq:sampling}
t \geq \Omega \left( \frac{ \log \frac{|\mathcal{Q}|}{n}}{ \log \frac{s}{n} } \right).
\end{equation}
where $|\mathcal{Q}| = m$ is the number of distinct queries. Breaking this barrier for an explicit data structure is a holy-grail problem, with connections to major open questions in branching programs and circuit complexity \cite{miltersen_data_1998, dvir_static_2019, viola_lower_2018}.

In summary, it was previously unknown if any linear problem was truly hard against non-linear data structures. For decades, it was not ruled out that highly efficient non-linear data structures (e.g., with $t = O( 1 )$ and $s = |\cQ|^{0.1}$) might exist for all linear problems.

\subsection{Further Connections} 

\subsubsection{Multiphase Program: Communication versus Cell-Probe}

A more recent motivation for our work is its connection to the communication version of the Multiphase Conjecture \cite{patrascu_mihai_towards_2010, thorup_mihai_2013, braverman_communication_2022}, a problem in 3-player Number-on-Forehead (NOF) communication.

\begin{definition}[Multiphase Communication Game \cite{patrascu_mihai_towards_2010}] The 3-player Number-On-Forehead (NOF) communication game is defined as follows
    \begin{itemize}
        \item The inputs are distributed as follows: Merlin receives the linear operator $M \in \{0,1\}^{m \times n}$ and the vector $X \in \{0,1\}^n$. For a given query index $i \in [m]$, Alice receives $M$ and $i$, while Bob receives $T$ and $i$.
        \item Merlin sends a single-shot message of length $m^{0.99}$ to Bob. 
        \item Alice and Bob proceed in standard two-party communication to output $\ip{M_i, X}$ \footnote{Originally, the conjecture was phrased in Merlin sending a message of length $o(m)$, then Alice and Bob proceeding to output $\DISJ(M_i,X)$. However, analogous implications hold from this slightly weaker conjecture} for any given $i \in [m]$.
    \end{itemize}
\end{definition}
The Multiphase Conjecture (Communication Version) then states that for some $m = \poly(n)$, Alice and Bob then must communicate $n^{\eps}$ for some constant $\eps > 0$. The original motivation in \cite{patrascu_mihai_towards_2010} behind the conjecture is its consequence in dynamic cell-probe lower bound. 
\begin{definition}[Multiphase Problem] Consider the following explicit dynamic data structure problem
    \begin{itemize}
    \item $M$ is given as a pre-processing input. Pre-process the data structure using $|M| \cdot t_u$ time.
    \item Then $X$ as a sequence is given as updates, updating the data structure using $|X| \cdot t_u$ total time.
    \item For any $i \in [m]$, the data structure must be able to output $\ip{M_i,X}$ in $t_q$ time.
    \end{itemize}
\end{definition}
The key observation in \cite{patrascu_mihai_towards_2010} is that a polynomial lower bound for the Multiphase Game would imply a polynomial lower bound on $\max \{ t_u, t_q \}$ for the Multiphase Problem, which is also known as the Multiphase Conjecture (Cell-Probe Version). The word ``Multiphase Conjecture" has been used interchangeably to denote either the communication version of the conjecture or the cell-probe version of the conjecture. 


To avoid any confusion, we will denote the communication conjecture as Multiphase Conjecture (Communication), the dynamic data structure conjecture as Multiphase Conjecture (Cell-Probe), and making progress on either one of these conjectures as the Multiphase Program. The main goal for the rest of the section is the following: (i) group the previous results depending on which conjecture they make a progress on; then (ii) illustrate why the Multiphase Conjecture (Communication) is a {\bf different ball game} compared to the Multiphase Conjecture (Cell-Probe).

\paragraph{Previous Results on the Multiphase Program}
Despite its significance, the Multiphase Program had limited progress for the first 10 years after its inception namely \cite{chattopadhyay_little_2012, clifford_new_2015, brody_adapt_2015}. \cite{clifford_new_2015, brody_adapt_2015} tackle the Multiphase Conjecture (Cell-Probe), giving a lower bound of $\max \{ t_u, t_q \} \geq \Omega ( \log n )$. But it should be noted that their bounds do not carry over to the Multiphase Conjecture (Communication) due to the limitation of the technique involved.

There are results on the Multiphase Conjecture (Communication) as well. \cite{chattopadhyay_little_2012} first tried to tackle the Multiphase Conjecture (Communication), but their argument only works for a very restricted model of communication. Only recently \cite{ko_adaptive_2020} developed a tool to tackle the Multiphase Conjecture (Communication) for two rounds of communication between Alice and Bob. This was extended to a more general class of functions by \cite{dvorak_lower_2020}. 

\paragraph{Separation between Communication and Cell-Probe} 

An intuitive high-level explanation on why tackling the Multiphase Game is a harder task is the following. Recall that in the reduction due to \cite{patrascu_mihai_towards_2010}, Merlin assumes the role of updates, and Alice assumes the role of query algorithm. But in the Multiphase Game, (a) Merlin only needs to send cells written by the update algorithm, as Merlin knows both $M$ and $X$. We do not charge for reading the cells during the update; (b) Alice is given the linear operator $M$. Thus, communication is needed for Alice to learn about $X$ to compute $\ip{M_i,X}$. On the other hand, if we were to directly approach the Multiphase Problem and give a cell-probe lower bound, the counting arguments in \cite{clifford_new_2015, brody_adapt_2015} crucially rely on the fact that (a) the update algorithm has no prior knowledge on $M$, thus must probe into pre-processed cells to learn $M$; (b) the query algorithm has no prior knowledge on both $M$ and $X$. The query algorithm must learn about both $M$ and $X$ by probing into the pre-processed and updated cells. 

An explicit example of the separation between the two models exists as well. There exists an explicit function that is easy in the communication model while hard in dynamic cell-probe model, separating the two models. Consider the setting where each $M_i$ has a single non-zero entry. This is so-called indexing problem, as each $M_i$ only asks which coordinate of $X$ to output. \cite{brody_adapt_2015} (Theorem 1) shows that if $t_u \leq o ( m )$, $t_q \geq \Omega ( n )$ for so-called non-adaptive query. This easily translates to $t_q \geq \Omega ( \log n )$ for general query, by simulating all possible adaptive queries in a single non-adaptive query. 

We will now show why we cannot expect to show such statement for the indexing problem in the communication model. First of all, Merlin's message is too long. We cannot show an analogous statement in the communication model, as $n \cdot t_u$, which would be the length of Merlin's message, is already larger than $m$. If Merlin is allowed to send a message of length $m$, Merlin can simply write the answers for all possible queries. Bob then only needs to a single bit to announce the answer. Notice that this stems from the difference listed as (a). 

Even if we relax the condition on the update to say $n \cdot t_u \leq o (m)$, thereby Merlin sending a message of length $o(m)$, we still run into separation, as we cannot give $t_q \geq \Omega ( \log n )$ for indexing problem. Such a result would contradict \cite{drucker_limitations_2012} which combined with the argument from \cite{viola_lower_2018} gives a communication protocol with $O(k^3)$ communication between Alice and Bob while Merlin only gives an advice of length $O(n/k)$ when $m = O(n)$.\footnote{One could also obtain $O(k^2)$ communication with advice of length $O \left( \frac{n \log n}{k \log \log n} \right)$} This shows that even if we give (a) for free, (b) is still crucial, separating the two model.

Now that we have established the fact that the Communication model is provably stronger than its dynamic cell-probe analogue, we would like to discuss evidences on how hard can it be, which has direct connection to our problem. Ko and Weinstein \cite{ko_adaptive_2020} observed that the Multiphase Communication Game can simulate a static cell-probe data structure for some random linear operator $M$ over the input $X$ -- the main problem of our interest.

Merlin sends all the contents of the pre-processing. Alice then simulates the query using the two-party communication between Bob. This is possible as Alice knows $M$, the linear operator in question. Therefore, resolving the Multiphase Conjecture (Communication) would imply Jukna-Schnitger Conjecture in its cell-probe analogue. In particular, the Multiphase Conjecture (Communication) implies a static cell-probe lower bound of the form 
\begin{equation}
    t \geq \left( \frac{|\mathcal{Q}|}{s} \right)^{c}
\end{equation}
for some small constant $c > 0$ against arbitrary linear operators with $|\cQ|$-many outputs. Note that this is a super-exponential improvement over \pref{eq:sampling}. Therefore, Ko and Weinstein (see Section 5 of \cite{ko_adaptive_2020}) phrased the connection as a {\bf major setback} for the full resolution of Multiphase Conjecture (Communication), rather than adding the stakes for the Multiphase Conjecture (Communication). We emphasize that no such connection arises for the Multiphase Conjecture for the Multiphase Conjecture (Cell-Probe). 

\subsubsection{Other Connections}

Our work also connects to index coding with side information~\cite{bar-yossef_index_2011}, the complexity of error-correcting codes~\cite{gal_tight_2013}, and function inversion in cryptography~\cite{hofheinz_function-inversion_2019}. These diverse connections all stem from the same high-level question: can non-linear computation help compute linear functions?

\subsection{Our Result}

Our main contribution is a query time lower bound that breaks the logarithmic barrier for a random linear operator:
\begin{equation} \label{eq:trade-off_ours} 
t \geq \Omega \left( \log \frac{|\mathcal{Q}|}{s} \right) 
\end{equation}
This provides the first formal evidence supporting the half-century-old intuition that non-linear pre-processing offers little advantage for most linear operators. For example, when space $s=|\cQ|^{0.99}$ (storing nearly all answers), previous bounds gave only a trivial $t \geq \Omega(1)$. Our bound gives a much stronger $t \geq \Omega (\log|\cQ|)$. This affirms the cell-probe version of Jukna-Schnitger's Conjecture when the number of queries is super-polynomial.

Formally, we prove the following statement:
\begin{theorem}[Main] \label{thm:informal_main} 
For every $t, s, m \geq \omega(1)$ and a collection $\cS \subset \Field_2^n$ such that $m \geq \omega ( s \cdot 2^{5 t} )$, and $\log |\cS| \geq 10^4 \cdot t \log s$,  there must exist some $m$ linear functions from $\cS$ say $\cQ$ that does not admit a data structure using $s$-space and $t$-probes per query. 
\end{theorem} 
\noindent When applied to the set of all linear functions $\cS = \Field_2^n$, this theorem yields our main trade-off.

\subsubsection{Implications}

Our result has direct implications in three key areas:

\paragraph{Static Data Structures}
We establish the existence of ``hard" linear operators, providing the first non-trivial lower bounds against arbitrary non-linear pre-processing in the high-space regime when $s$ is $|\cQ|^{0.99}$.
This is a significant improvement over any existing arguments, which fail in this setting.

\paragraph{Circuit Complexity}

Our data structure bounds translate to new circuit lower bounds. Consider Valiant's depth-2 circuit with {\bf arbitrary} gates \cite{goos_graph-theoretic_1977} computing a linear operator $x \mapsto M x$ with $M \in \Field_2^{m \times n}$, $m^{0.99}$-gates in the middle layer, and at most $\tau$-wires per output gate. Our result directly implies that if $t \leq o ( \sqrt{n} )$ and $m \geq 2^{ O(t) }$, there must exist a linear operator $M$ that requires $\tau \geq \Omega(t)$. This partially resolves an open problem of Jukna-Schnitger~\cite{jukna_min-rank_2011, jukna_boolean_2012} on whether arbitrary gates can help in computing an arbitrary linear operator for a super-polynomial number of outputs.

Furthermore, if we turn to a general depth $d$ circuit with arbitrary gates (unbounded fan-in), we can also replace the inverse Ackermann-type lower bounds due to \cite{gal_tight_2013} on linear operators with poly-logarithmic decay in depth $d$.
Our bound also supersedes the best-known explicit bound, which are not linear operators, due to \cite{raz_lower_2003, cherukhin_lower_2008, hirsch_lower_2008} when $d \geq 4$. See \pref{sec:viola} for details.

\paragraph{Multiphase Communication Game}
A general way to phrase our result is a progress in the Multiphase Conjecture (Communication). 
While the notation of our proof is used specifically to handle the static data structure lower bound, 
a careful examination of our proof shows that our result indeed resolves the Multiphase Conjecture (Communication)~\cite{patrascu_mihai_towards_2010} when the number of possible queries $m$ is super-polynomial. 

Since any progress on the Multiphase Conjecture (Communication) implies a corollary for the Multiphase Conjecture (Cell-Probe), our lower bound then implies the following corollary for the Multiphase Problem.
\begin{corollary} \label{cor:multiphase}
    For the Multiphase Problem, if the total update time $n \cdot t_u \leq O( m^{0.99})$ (i.e. slightly less than writing down answers to all possible queries) then the query time must be $t_q \geq \Omega ( n^{1/2} )$
\end{corollary}
\noindent We suspect further reductions to various explicit dynamic data structure problems from \pref{cor:multiphase}. Note that a trivial data structure exists for the above explicit dynamic data structure instance, achieving $t_q = O(n)$ by simply reading $M_i$ and $X$ to compute $\ip{M_i, X}$. This is important because one could have otherwise devised a dynamic data structure problem which selects arbitrary functions $f_1, \ldots, f_m : \{ 0 , 1 \}^n \to \{0,1\}$, and outputs $f_i ( X )$ at the query phase. 

Using Miltersen’s counting argument \cite{miltersen_bit_1993}, one can readily demonstrate the existence of functions $f_1, \ldots, f_m$ such that, if the total update time $n \cdot t_u = O(m^{0.99})$, then $t_q \geq \Omega(n^{0.99})$. However, mainly due to requiring an arbitrary function $f_i$, which requires $2^n$-bits to describe per query, this prohibits any trivial query algorithm of $t_q = O(n)$, as one needs to read about $f_i$, nor reduces to any interesting problem in dynamic data structure. \footnote{Another way to view this phenomenon is that our instance requires small linear bits ($n$-bits) per query, while Miltersen's counting argument requires exponential bits per query.}

\subsubsection{Technical Contribution}  

The main technical ingredient of our work is the extension of the analysis of Multiphase Game~\cite{ko_adaptive_2020, dvorak_lower_2020} to exponentially small advantage regimes using average min-entropy introduced in \cite{dodis_fuzzy_2008}. We suspect that our technique to be a critical component in fully resolving the Multiphase Program~\cite{patrascu_mihai_towards_2010, thorup_mihai_2013, braverman_communication_2022}. Analyzing the small advantage regime is crucial for data structure lower bounds, as demonstrated in recent breakthroughs in dynamic data structure lower bounds \cite{larsen_crossing_2020, larsen_super-logarithmic_2025}.

The main technical contribution in \cite{ko_adaptive_2020} shows that after Merlin's message, if Alice and then Bob each speak once and correctly output $\DISJ(M_i, X)$ with $s = o(m)$, then Alice and Bob must communicate $\Omega( \sqrt{n} )$ bits. \cite{dvorak_lower_2020} extended the result to $\ip { M_i , X }$ with $\Omega (n)$ bits of communication.

The key open problem posed in \cite{ko_adaptive_2020} is to extend the previous result to the setting where Alice, Bob then Alice speaks (i.e. one extra round compared to \cite{ko_adaptive_2020, dvorak_lower_2020}). This has connections to the fundamental question of linear vs. non-linear circuits as remarked in \cite{ko_adaptive_2020}. We reduce the general setting with $t$-rounds of communication to a setting where Alice and then Bob speak with reduced correctness, via following observation.

Alice does the following: she picks a single random $t$-probe decision tree path. Then Alice sends the path to Bob. Bob checks if the path is consistent with $U$, Merlin's message. If yes, Bob outputs the same output as the end of the path. Otherwise Bob simply makes a random guess about the output. The main observation is that Bob would say yes with $2^{-t}$ probability. If Bob says Yes, then the protocol is correct. Otherwise, we get 0 advantage. Therefore, this would result in $\Omega( 2^{-t} )$ advantage over random guessing.

We then focus on the setting in which Alice and then Bob each speak once, with Bob guessing the value of $\ip { M_i , X }$ with a small advantage over a random guessing. That is, we consider {\bf the small advantage regime}. Unfortunately, previous arguments \cite{ko_adaptive_2020, dvorak_lower_2020} incurred a constant blowup in error, making them insufficient for any meaningful bound for small advantage regimes. In fact, as observed in two-party Disjointness \cite{braverman_information_2013}, analyzing small advantage regime often requires different tools from that of small error regime. Our technical novelty here is to leverage the average min-entropy framework introduced in \cite{dodis_fuzzy_2008} instead of standard information theoretic arguments using KL-divergence, and carefully choose random variables that fit the framework of \cite{ko_adaptive_2020}. Average min-entropy allows us to naturally bound the $\ell_2$-norm of the probability distribution, which then can be then used to argue about the underlying discrepancy. 
\section{Preliminary} 

\subsection{Information Theory}

In this section, we provide the necessary background on information theory and information complexity that are used 
in this paper. For further reference, we refer the reader to \cite{cover_elements_2006}.

\begin{definition}[Entropy]
	The entropy of a random variable $X$ is defined as 
	\begin{equation*}
	H(X) := \sum_{x} \Pr[X=x] \log \frac{1}{\Pr[X=x]}.
	\end{equation*}
	Similarly, the conditional entropy is defined as
	\begin{equation*}
	H(X|Y) := \E_{Y} \left[ \sum_{x} \Pr[X = x | Y = y] \log \frac{1}{\Pr[ X = x | Y = y]} \right].
	\end{equation*}
\end{definition}

\begin{fact}[Conditioning Decreases Entropy] \label{fact:conditioningentropy}
For any random variable $X$ and $Y$
\begin{equation*}
    H(X) \geq H(X|Y)
\end{equation*}
\end{fact}

\noindent With entropy defined, we can also quantify the correlation between two random variables, or how much information one random variable conveys about the other.

\begin{definition}[Mutual Information]
	Mutual information between $X$ and $Y$ (conditioned on $Z$) is defined as 
	\begin{equation*}
	I( X; Y | Z) := H(X|Z) - H(X|YZ).
	\end{equation*}
\end{definition}

\noindent Similarly, we can also define how much one distribution conveys information about the other distribution.

\begin{definition}[KL-Divergence]
	KL-Divergence between two distributions $\mu$ and $\nu$ is defined as
	\begin{equation*}
	\KL{\mu}{\nu} := \sum_{x} \mu(x) \log \frac{\mu(x)}{\nu(x)}.
	\end{equation*}
\end{definition}

\noindent To bound mutual information, it suffices to bound KL-divergence, due to the following fact.

\begin{fact}[KL-Divergence and Mutual Information] \label{fact:KL_Mutual}
	The following equality between mutual information and KL-Divergence holds 
	\begin{equation*}
	I(A;B|C) = \E_{B,C} \left[ \KL{ A|_{B=b, C=c} }{ A|_{C=c} } \right].
	\end{equation*}
\end{fact}


\begin{fact}[Pinsker's Inequality] \label{fact:pinsker}
For any two distributions $P$ and $Q$, 
\begin{equation*}
\norm{ P - Q }_{TV} = \frac{1}{2} \norm{ P - Q }_1 \leq \sqrt{ \frac{1}{2 \log e} D( P || Q) }	
\end{equation*}
\end{fact}

\noindent We also make use of the following facts on Mutual Information throughout the paper.

\begin{fact}[Chain Rule] \label{fact:chainrule}
For any random variable $A,B,C$ and $D$  
\begin{equation*}
    I(AD;B|C) = I(D;B|C) + I(A;B|CD).
\end{equation*}
\end{fact}

\begin{fact} \label{fact:chainrule1}
For any random variable $A,B,C$ and $D$, if $I(B;D|C) = 0$
\begin{equation*}
    I(A;B|C) \leq I(A;B|CD).
\end{equation*}
\end{fact}
\begin{proof} By the chain rule and non-negativity of mutual information, 
\begin{align*}
I(A;B|C) \leq I(AD;B|C) = I(B;D|C) + I(A;B|CD) = I(A;B|CD).
\end{align*}
\end{proof}

\begin{fact}\label{fact:chainrule2}
For any random variable $A,B,C$ and $D$, if $I(B;D|AC) = 0$
\begin{equation*}
    I(A;B|C) \geq I(A;B|CD).
\end{equation*}
\end{fact}
\begin{proof} By the chain rule and non-negativity of mutual information,  
\begin{align*}
I(A;B|CD) \leq I(AD;B|C) = I(A;B|C) + I(B;D|AC) = I(A;B|C).
\end{align*}
\end{proof}

\subsection{$\ell_2$-norm of a distribution}

Another good measure of the randomness of a distribution is its $\ell_2$-norm. The more ``spread out" a distribution is, the smaller its $\ell_2$-norm. We introduce the following definitions to argue about the $\ell_2$-norm of a distribution.

\begin{definition} We define the renyi entropy $H_2 (A)$ and min-entropy $H_\infty (A)$ as
    \begin{align*}
        H_2 (A) := - \log \left( \sum_{a} \Pr[ A = a]^2 \right) \\
        H_\infty (A) := - \log \left( \max_{a} \Pr[ A = a]  \right)
    \end{align*}
\end{definition}

\begin{fact}[Renyi Entropy] \label{fact:renyi}
Let $A$ be a random variable. Then
	\begin{equation*}
		H ( A ) \geq H_2 (A) \geq H_\infty (A) 
	\end{equation*}
 In particular, for any fixed $b$ we have
\begin{equation*}
H_2 ( A| B = b ) \geq H_\infty ( A | B = b)	
\end{equation*}
\end{fact}

We use the following simple claim to argue about $\ell_\infty$-norm of a distribution when its KL-divergence with the uniform distribution is small.

\begin{claim} \label{cl:lowL2}
    Let $\cD$ be a distribution and $\cU$ a uniform distribution over some $\cS \subset \{0,1\}^n$. Then if $D( \cD || \cU ) < t$ with $t > 1$, then for every $\alpha > 2$ there exists an event $E$ such that 
    \begin{align*}
        & \cD ( E ) \geq 1 - \frac{1}{\alpha} \\ 
        & H_\infty ( \cD |_{E} ) \geq  \log |\cS| - 2 \alpha t  
    \end{align*}
\end{claim}
\begin{proof}
    We partition $X$ depending on $\cD (X)$. Let $E$ denote the set of $X$ such that
    \begin{equation*}
        \log \frac{\cD(X)}{\cU (X)} < \alpha t 
    \end{equation*}
    Since $D( \cD || \cU ) < t$, by Markov's inequality $1 - \cD (E) < 1 / \alpha$. With $\alpha > 2$, for any $X$ that is in the support of $\cD|_{E}$, we have 
    \begin{equation*}
        \cD|_{E} (X) \leq 2 \cdot \cD( X) \leq 2^{\alpha t + 1} \cdot \cU (X ) = 2^{- \log |\cS| + \alpha t + 1} \leq 2^{ -  \log |\cS| + 2 \alpha t}
    \end{equation*}
    which then gives $H_\infty ( \cD |_{E} ) \geq   \log |\cS|  - 2 \alpha t $.
\end{proof}

We prove the following lemma to bound the $\ell_2$-norm of a distribution.
\begin{lemma} \label{lem:collision}
    Let $A, B$ be random variables where $B$ has at most $2^\lambda$ possible values. Then 
    \begin{equation*}
        \E_{b \sim B} \left[ \sum_{a} \Pr [ A = a | B = b]^2 \right] \leq 2^{- H_\infty (A) + \lambda}
    \end{equation*}
\end{lemma}

We use the following lemma on ``average" min-entropy.

\begin{definition}[Average Min-Entropy]
	\begin{equation*}
		\tilde{H}_\infty ( A | B ) = - \log \left( \E_{b \sim B} \left[ \max_{a} \Pr [ A = a | B = b ] \right] \right) = - \log \left( \E_{b \sim B} \left[ 2^{- H_\infty (A | B = b ) } \right] \right)	
	\end{equation*}
\end{definition}

This leads to the following proposition
\begin{proposition} \label{prop:l-2-min}
    \begin{equation*}
        \E_{b \sim B} \left[ \sum_{a} \Pr [ A = a | B = b]^2 \right] \leq 2^{- \tilde{H}_\infty (A | B )}
    \end{equation*}
\end{proposition}
\begin{proof}
    \begin{equation*}
        \sum_{a} \Pr [ A = a | B = b]^2 = 2^{ - H_2 (A | B = b ) } \leq 2^{ - H_\infty (A | B = b ) } 
    \end{equation*}
    Therefore, 
    \begin{equation*}
        \E_{b \sim B} \left[ \sum_{a} \Pr [ A = a | B = b]^2 \right] \leq \E_{b \sim B} \left[ 2^{- H_\infty (A | B = b )} \right] = 2^{ - \tilde{H}_\infty (A | B ) }
    \end{equation*}
    where the last equality holds from the definition of average min-entropy.
\end{proof}

\begin{lemma}[Lemma 2.2 of \cite{dodis_fuzzy_2008}] \label{lem:dors}
	Let $A,B$ be random variables. Then if $B$ has at most $2^{\lambda}$ possible values, then 
	\begin{equation*}
		\tilde{H}_\infty ( A | B ) \geq \tilde{H}_\infty (A, B ) - \lambda \geq H_\infty ( A ) - \lambda.
	\end{equation*}
\end{lemma}

We are now ready to prove \pref{lem:collision} 

\begin{proofof}{\pref{lem:collision}}
First, observe that \pref{lem:dors} implies 
\begin{equation*}
    - \log \left( \E_{b \sim B} \left[ 2^{- H_\infty (A | B = b ) } \right] \right)	\geq H_\infty (A) - \lambda
\end{equation*}
which is equivalent to (by switching sides and taking as exponents)
\begin{equation*}
    2^{- H_\infty (A) + \lambda } \geq \E_{b \sim B} \left[ 2^{- H_\infty (A | B = b ) } \right]
\end{equation*}
Now using \pref{fact:renyi}, you get 
\begin{equation*}
\E_{b \sim B} \left[ 2^{- H_\infty (A | B = b ) }  \right]	\geq \E_{b \sim B} \left[ 2^{- H_2 (A | B = b ) } \right] = \E_{b \sim B} \left[ \sum_{a} \Pr [ A = a | B = b ]^2 \right].
\end{equation*}
which is the desired inequality.
\end{proofof}

Using the $\ell_2$ norm, the following claim (so-called Lindsey's Lemma) bounds the discrepancy of the Hadamard matrix under a product distribution.
\begin{claim}[Lindsey’s Lemma] \label{cl:hadamard}
Let $H$ be a Hadamard matrix. Let $P$ and $Q$ be distributions. Then 
\begin{equation*}
    P^{T} H Q \leq \norm{ P }_2 \norm{ Q }_2 \cdot 2^{n/2}
\end{equation*}
\end{claim}
\begin{proof}
    Recall that the operator norm of the Hadamard matrix is exactly 
    \begin{equation*}
        \norm{H}_2 = 2^{n/2}
    \end{equation*}
    Then by the definition of the operator norm,
    \begin{equation*}
        \frac{P^{T} H Q }{\norm{ P }_2 \norm{ Q }_2 } \leq \norm{H}_2 = 2^{n/2}
    \end{equation*}
    Rearranging the inequality, we obtain the desired claim.
\end{proof}

\section{Main Proof}

In this section, we prove the following main theorem.

\begin{theorem}[Main] \label{thm:main}
    For every $t, s, m \geq \omega(1)$ and a collection $\cS \subset \Field_2^n$ such that $m \geq \omega ( s \cdot 2^{5 t} )$, and $\log |\cS| \geq 10^4 \cdot t \log s$, there must exist some $m$ linear functions from $\cS$ say $\cQ$ that does not admit a data structure using $s$-space and $t$-probes per query. 
\end{theorem}



\subsection{Compression and Random Process}

We prove this by contradiction. We will start with the following assumption. Suppose for any $M_1, \ldots, M_m \in \cS \subseteq \Field_2^n$ (these are fixed and publicly known), there exists a data structure under the cell-probe model which pre-processes any given $X \in \Field_2^n$ using $s$-space and answers $\chi_{M_i} ( X ) := ( - 1 )^{\ip{M_i,X}}$ for any given $i \in [ |\cQ| ] = [m]$ using $t$-probe. This implies the following procedure exists in the cell-probe model.

\begin{Protocol}
    \begin{enumerate}
        \item The querier is given $\vec{M} = M_1, \ldots, M_m$. Given $i \in [m]$, the querier makes $t$ queries to $U$, using a decision tree $\cT$ of depth $t$.
        \item At the leaf node of the deicision tree $\cT$, the querier outputs the final guess in $\{ \pm 1 \}$ which is guaranteed to be $\chi_{M_i} ( X )$.
    \end{enumerate}
    \caption{ Static Data structure with adaptive probes }
\end{Protocol} 

We now argue that an effective static data structure implies an too-good-to-be true communication process in the so-called Multiphase Communication Game. While \Patrascu~\cite{patrascu_mihai_towards_2010} considered computing the Disjointness of $M_i$ and $X$, here we consider computing the inner product over mod 2 between $M_i$ and $X$.

The above data structure implies the existence of the following $(t+1)$-round communication process (between 3 players) under the independent uniform distribution for all $M_1, \ldots, M_m \in \cS \subseteq \Field_2^n$ and uniform distribution for $X \in \Field_2^n$ where Merlin sends a single-shot message of length $s$, $U(M_1, \ldots, M_m, X)$ to Bob. Then Alice and Bob proceed in $t$-round communication to compute $\chi_{M_i} ( X )$. By the correctness of the underlying static data structure, this process correctly outputs the inner product of $M_i$ and $X$ for any values of $i \in [m]$.

\begin{Protocol}
    \begin{enumerate}
        \item Alice holds $M_1, \ldots, M_m$ and $i \in [m]$. Bob holds $X$ and $i \in [m]$. Merlin holds $M_1, \ldots, M_m, X$.
        \item Merlin sends $U (M_1, \ldots, M_m, X)$ to Bob using $s$ bits.
        \item Alice sends a location $q \in [s]$ (given in her decision tree $\cT$) using $\log s$-bits to Bob. Bob replies with the queried bit $U_q$.
        \item Repeat the Step 3 for $t$ rounds.
        \item Alice announces $\chi_{M_i} ( X )$. 
    \end{enumerate}
    \caption{ $t + 1$-round Communication Process \label{prot:multiphase}} 
\end{Protocol} 

We denote this as a process, rather than a protocol, to contrast with the usual two-party communication protocol. Due to the side information $U$ that Bob has on Alice's input, the usual techniques for two-party communication complexity (namely cut-and-paste or rectangular property) do not hold. Therefore, we call this a communication process instead of a communication protocol.

\paragraph{Compression Overview}
The hope is to show that if $s$ is sufficiently small then $t$ must be large for such a $(t+1)$-round communication process. Unfortunately, current techniques are insufficient to provide a complete lower bound for such a general process. Instead, we would like to ``compress" the communication between Alice and Bob with a loss in parameters, resulting in a 3-round communication process between Merlin, Alice, and Bob.
We can manage this setting from the technique developed in \cite{ko_adaptive_2020}.

The high-level intuition is for Alice to randomly sample a path in her decision tree, along with the final answer. Bob simply checks whether or not queried bits indeed matches with $U$. If it matches, Bob announces the answer as given in Alice's message. Otherwise, Bob simply makes an independent random guess. We formally state the compression as following. 

\paragraph{Formal Compression}
Here is the formal description of our compressed 3-round communication process. 


\begin{Protocol}
    \begin{enumerate}
        \item Merlin sends $U (M_1, \ldots, M_m, X)$ to Bob using $s$ bits.
        
        \item Let $P_i \in ( [s], \{0,1\} )^{t} \times \{ \pm 1 \} $, namely query $i$'s (i) the sequence of the addresses probed by the probing algorithm; (ii) the probed result and; (iii) the final answer (i.e. the contents of a path in the decision tree used by Alice) Note that there are at most $2^t$ many possible $P_i$. Alice picks a $P_i = p_i$ such that
        \begin{equation*}
            p_i := \argmax \Pr_{X} [ P_i = p_i | \vec{M} = \vec{m}]
        \end{equation*}
        then sends to Bob using $t \log s + t + 1$ bits.
        
        \item Bob checks if $P_i$ agrees with $U$. If yes, set $B_i$ as 1. Output $Z_i$, which is the notation for Bob's guess of $\chi_{M_i} (X)$, as the final output value in $P_i$.  Otherwise, set $B_i$ as 0. Output $Z_i$ as $\pm 1$ each with probability $1/2$ independently at random. 
    \end{enumerate}
    \caption{ Communication Process between Alice and Bob for Adaptive Probe \label{prot:AliceBob_adaptive} } 
\end{Protocol}

\begin{remark} \label{rem:multiphase}
    Observe that our argument works against even a general three party communication, where after Merlin's message, Alice and Bob proceed to communicate $t$ bits using $t$-rounds of communication (i.e. a single bit per round). As long as Alice's bits are a function of previous transcript, and her input $(M_1, \ldots, M_m, i)$; and Bob's bits are a function of previous transcript, his input $(X, i)$ and Merlin's message, the above compression scheme works. Bob simply needs to check if the response chosen by the path is consistent with his actual response.
\end{remark}
Due to \pref{rem:multiphase}, the rest of proof also works against arbitrary Number-On-Forehead three party communication, thereby attacking the Multiphase Conjecture (Communication). However, for conciseness, we stick to above notations against static data structure lower bound.

\paragraph{Random Process from \pref{prot:AliceBob_adaptive}}

Here is the plan for the remainder of the proof. Given a too-good-to-be true compressed communication process (\pref{prot:AliceBob_adaptive}), we would like to ``extract" a random process $Z$ from our compressed communication process which would violate some combinatorial property. We will show that some choice of $Z$ (after some conditioning $E$ which is introduced due to technical reasons) satisfies three properties simultaneously: conditioned on $Z$, large average Min-Entropy on $M_{i}$ and $X$, low correlation between $M_{i}$ and $X$, while achieving a good advantage over randomly guessing $\chi_{M_{i}} (X)$. Then in \pref{sec:combinatorial}, we will then show that such a choice of $Z$ cannot exist, resulting in a contradiction.

Given the formal description of our compression, we formally state our random process $Z$. Pick a sequence of random distinct numbers $\cP$ of length $\ell \leq m / 100$ in $[m]$, an instance of which we will denote as $(\rho_1, \ldots, \rho_\ell)$, and random $J \in [\ell]$, an instance of which we will denote as $j$. Then we write 
\begin{equation*}
    Z := \cP, J, P_{ \cP_{< J} }, M_{ \cP_{<J} }, P_{ \cP_{J} }, B_{\cP_{J}}, Z_{\cP_{J}}
\end{equation*}
which is the set of random ``path" of coordinates along with Alice's compressed messages and $S_i$'s along the path. We also include Bob's output (whether or not $P_{ \cP_{J} }$ agrees with $U$) for the particular $J$ of our choosing. For brevity, we write $Z_{Alice}$ for
\begin{equation*}
    Z_{Alice} := \cP, J, P_{ \cP_{< J} }, M_{ \cP_{<J} }, P_{ \cP_{J} }
\end{equation*}
that is $Z$ without $B_{\cP_{J}} $ (i.e. Bob's message) and the final guess $Z_{\cP_{J}}$ due to Bob's message. We emphasize here that $Z_{Alice}$ is completely determined by $M_1, \ldots, M_m = \vec{M}$ and independent randomness $\cP, J$, and has no dependence on $X$ (Bob's input) which is a crucial property of $Z_{Alice}$ and why the variable is named so.

\subsubsection{Large Average Min-Entropy}

In this section, using the communication process $Z$, first, we argue on large average Min-Entropy for $M_{\cP_J}$ and $X$ conditioned on $Z$. As a comparison, we remark that \cite{ko_adaptive_2020, dvorak_lower_2020} extract $Z = z$ with a small KL-divergence which is directly implied by small mutual information between $Z$ and respective $M_{\cP_J}$ and $X$.

We start with the following lemma which directly follows from the technique used in \cite{ko_adaptive_2020}.
\begin{lemma} \label{lem:S}
    If $|\cP| = \ell = |\cQ|/100$ then
    \begin{equation*}
        I ( M_{\cP_J} ; Z_{Alice} ) = \E_{z_{Alice}} \left[ D( M_{\cP_J} |_{Z_{Alice} = z_{Alice} } || M_{\cP_J} ) \right]  \leq 3 t \log s
    \end{equation*}
\end{lemma}
We attach the proof of \pref{lem:S} in the appendix for completeness, as it is exactly the same as in \cite{ko_adaptive_2020}. What \pref{lem:S} guarantees is $Z_{Alice}$ such that conditioning on which has low divergence with original uniform distribution on $M_{\cP_J}$.

The following is the main lemma of this section, which combined with \pref{lem:S}, we can extract large average Min-Entropy part for $M_{\cP_J}$ and $X$.

\begin{lemma}[Large Average Min-Entropy] \label{lem:l2S}
    For any setting of fixed $Z_{Alice} = z_{Alice}$ such that 
    \begin{equation*}
        D( M_{\cP_J} |_{Z_{Alice} = z_{Alice}} || M_{\cP_J} ) \leq \beta
    \end{equation*} 
    with $\beta > 2$, there exists an associated event $E$ with 
    \begin{align*}
        & \Pr \left[ E | Z_{Alice} = z_{Alice} \right] \geq 3/4 \\
        & H ( E | Z_{Alice}, M_{\cP_J} ) = 0 \\
        & \tilde{H}_\infty ( M_{\rho_j} | Z_{\rho_j} B_{\rho_j}, E = 1, Z_{Alice} = z_{Alice} )  + \tilde{H}_\infty ( X | Z_{\rho_j} B_{\rho_j}, E = 1, Z_{Alice} = z_{Alice} ) \\
        & \geq  n + \log |\cS| - 10 \beta       
     \end{align*}
\end{lemma}

Towards proving the lemma, we first prove the following two propositions.
\begin{proposition} \label{prop:El2}
    Suppose
    \begin{equation*}
        D( M_{\cP_J} |_{Z_{Alice} = z_{Alice}} || M_{\cP_J} ) \leq \beta.
    \end{equation*}
    Then there exists an event $E$ with $\Pr [ E | Z_{Alice} = z_{Alice} ] \geq 3/4$ such that 
    \begin{align*}
        & H ( E | Z_{Alice}, M_{\cP_J} ) = 0 \\ 
        & H_{\infty} ( M_{\rho_j} | Z_{Alice} = z_{Alice}, E = 1 ) \geq \log |\cS| - 8 \beta
    \end{align*}
\end{proposition}
\begin{proof}
    Recall that \pref{cl:lowL2} gives an event $E$ where if the underlying distribution has low KL-divergence with the uniform distribution over $\cS$, conditioned on $E$, the distribution has large $H_\infty$. As
    \begin{equation*}
        D( M_{\rho_j} |_{Z_{Alice} = z_{Alice}} || M_{\rho_j} ) \leq \beta.
    \end{equation*}
    with $M_{\cP_J}$ being the uniform distribution over $\cS$, we can apply \pref{cl:lowL2} to give an event $E$ which is completely determined by $Z_{Alice}$ and $M_{\cP_J}$ (as $E$ is determined by $M_{\cP_J}$ under fixed $Z_{Alice}$) such that 
    \begin{equation*}
        H_{\infty} ( M_{\rho_j} | Z_{Alice} = z_{Alice}, E = 1 ) \geq \log |\cS| - 8 \beta
    \end{equation*}
    with 
    \begin{equation*}
        \Pr [ E | Z_{Alice} = z_{Alice} ] \geq 3/4
    \end{equation*}
    by setting the appropriate parameter ($\alpha = 4$ in \pref{cl:lowL2}).
\end{proof}
The next proposition is $H_\infty$ bound for $X$. 
\begin{proposition} \label{prop:Tl2}
    For any setting of $Z_{Alice} = z_{Alice}$, and $E = 1$, we have
    \begin{equation*}
        H_{\infty} \left( X | Z_{Alice} = z_{Alice}, E = 1 \right) = n.
    \end{equation*}
\end{proposition}
\begin{proof}
    First, by our setting of $Z_{Alice}$ and $E$, which depends only on $\vec{M} = M_1, \ldots, M_m$ (Alice's input),
    \begin{equation*}
        I ( Z_{Alice}, E ; X ) \leq I ( \vec{M} ; X ) = 0. 
    \end{equation*}
    Therefore, for any setting of $Z_{Alice}$ and $E = 1$, we get
    \begin{equation*}
        H_{\infty} \left( X | Z_{Alice} = z_{Alice}, E = 1 \right) = n.
    \end{equation*}
\end{proof}

We are now ready for the proof of \pref{lem:l2S}

\begin{proofof}{\pref{lem:l2S}}
    Due to \pref{prop:El2}, there exists an event $E$, that is completely determined by corresponding $S_{\rho_j}$ such that
    \begin{align*}
        & H ( E | Z_{Alice}, M_{\cP_J} ) = 0 \\
        & \Pr[ E | Z_{Alice} = z_{Alice} ] \geq 3/4 \\
        & H_{\infty} ( M_{\rho_j} | Z_{Alice} = z_{Alice}, E = 1 ) \geq \log |\cS| - 8 \beta
    \end{align*}
    which further implies that, due to \pref{lem:dors}, as $Z_{\rho_j} B_{\rho_j}$ is $2$-bit,
    \begin{equation}
        \tilde{H}_{\infty} ( M_{\rho_j} | Z_{\rho_j} B_{\rho_j}, Z_{Alice} = z_{Alice}, E = 1 ) \geq \log |\cS| - 8 \beta - 2. \label{eq:s_avg_min}
    \end{equation}
    Due to \pref{prop:Tl2}, and \pref{lem:dors}, for any setting of $Z_{Alice}$ we have
    \begin{align}
        \tilde{H}_{\infty} \left( X | Z_{\rho_j} B_{\rho_j}, Z_{Alice} = z_{Alice}, E = 1 \right) \geq n - 2 \label{eq:t_avg_min}
    \end{align}
    Then adding \pref{eq:s_avg_min} and \pref{eq:t_avg_min}, we obtain our desired lemma as
    \begin{align*}
        &  \tilde{H}_{\infty} \left( X | B_{\rho_j}, Z_{Alice} = z_{Alice}, E = 1 \right) + \tilde{H}_{\infty} ( M_{\rho_j} | B_{\rho_j}, Z_{Alice} = z_{Alice}, E = 1 ) \\
        & \geq n + \log |\cS| - 8 \beta  - 4 \geq n + \log |\cS| - 10 \beta.
    \end{align*}
    where the last inequality holds from our assumption $\beta > 2$.
\end{proofof}

\subsubsection{Low Correlation} 

Finally, we show that the correlation between $M_{\cP_J}$ and $X$ is small in expectation over random $Z$ and conditioned on $E$, which essentially follows the analogous argument in \cite{ko_adaptive_2020} using the Chain Rule in Mutual Information~(\pref{fact:chainrule2}). 

\begin{lemma}[Low Correlation] \label{lem:corr}
    \begin{equation*}
        I ( M_{\cP_J} ; X | Z , E = 1) \leq \frac{2 |U| }{|\cP|} = \frac{2 s}{\ell} 
    \end{equation*}
\end{lemma}
\begin{proof}
    We use an analogous technique from \cite{ko_adaptive_2020}. First, note that
    \begin{equation*}
        I ( M_{\cP_J} ; X | Z, E = 1) \leq I ( M_{\cP_J}  ; U X | Z, E = 1) 
    \end{equation*}
    Now we plug in the definition of $Z$ as the concatenation of $Z_{Alice}, Z_{\cP_j} B_{\cP_j}$ along with $E = 1$. Then, for any fixed $Z_{Alice} = z_{Alice}$, we get
    \begin{align*}
        I ( M_{\cP_J}  ; U X  | Z_{Alice} = z_{Alice} , E = 1, Z_{\rho_j} B_{\rho_j} ) \leq I ( M_{\cP_J}  ; U X  | Z_{Alice} = z_{Alice} , E = 1, B_{\rho_j} )
    \end{align*}
    As $I ( Z_{\rho_j} ; U X |  B_{\rho_j}, Z_{Alice} = z_{Alice} , E = 1, M_{\cP_J} ) = 0$. If $B_{\rho_j} = 1$, $Z_{\rho_j}$ is completely determined from $z_{Alice}$. Otherwise it is an independent random variable. Then
    \begin{align*}
        & I ( M_{\cP_J}  ; U X  | Z_{Alice} = z_{Alice} , E = 1, B_{\rho_j}) \leq I ( M_{\cP_J}  ; U X | Z_{Alice} = z_{Alice}, E = 1 ) \\
        & \leq \frac{1}{\Pr[E= 1| Z_{Alice} = z_{Alice} ] } I ( M_{\cP_J}  ; U X | Z_{Alice} = z_{Alice} , E ) \leq 2 \cdot I ( M_{\cP_J}  ; U X | Z_{Alice} = z_{Alice}, E ) \\
        & \leq 2 \cdot I ( M_{\cP_J}  ; U X | Z_{Alice} = z_{Alice} )
    \end{align*}
    where the inequalities hold by \pref{fact:chainrule2} along with
    \begin{equation*}
         I ( M_{\cP_J}  ; B_{\rho_j} | Z_{Alice} = z_{Alice} , E = 1, U X ) \leq H ( B_{\rho_j} | Z_{Alice} = z_{Alice} , E = 1, U X ) = 0
    \end{equation*}
    and the properties of $E$ guaranteed by \pref{lem:l2S}, namely $\Pr[E= 1|Z_{Alice} = z_{Alice} ] \geq 3/4$ and 
    \begin{equation*}
        I ( E  ; U X | Z_{Alice} = z_{Alice}, M_{\cP_J} ) \leq H ( E | Z_{Alice} = z_{Alice}, M_{\cP_J} ) = 0.
    \end{equation*}

    Next, we bound $I ( M_{\cP_J}  ; UT | Z_{Alice} )$ term.  Note that $I ( M_{\cP_J} ; X | Z_{Alice} ) = 0$ as
    \begin{align*}
        I ( M_{\cP_J} ; X | Z_{Alice} ) \leq I ( M_{\cP_J} Z_{Alice} ; X ) \leq I ( \vec{M} ; X ) = 0
    \end{align*}
    Next, we consider $I ( M_{\cP_J}  ; U | Z_{Alice}, X)$ then plugging in the definition of $Z_{Alice}$, 
    \begin{align*}
        & I ( M_{\cP_J}  ; U | Z_{Alice}, X ) = I ( M_{\cP_J}  ; U | \cP, J, P_{\cP_J} , P_{ \cP_{<J} }, M_{\cP_{<J}}, X ) \\
        & \leq I ( P_{ \cP_{J} }, M_{\cP_J}  ; U | \cP, J, P_{ \cP_{<J} }, M_{\cP_{<J}}, X). 
    \end{align*}
    Then taking the expectation over the random coordinate $J$ (by standard direct-sum technique), as all other variables are chosen independently of $J$, we get
    \begin{align*}
        & I ( P_{ \cP_{J} }, M_{\cP_J}  ; U | \cP, J,  P_{ \cP_{<J} }, M_{\cP_{<J}}, X ) \\
        & = \frac{1}{\ell} \sum_{j =1}^{\ell} I ( P_{ \cP_{j} }, M_{\cP_J}  ; U | \cP, J = j, P_{ \cP_{< j} }, M_{\cP_{<J}}, T) \\
        & = \frac{1}{\ell} \sum_{j =1}^{\ell} I ( P_{ \cP_{j} }, M_{\cP_J}  ; U | \cP, P_{ \cP_{< j} }, M_{\cP_{<J}}, T) \\
        & = \frac{I ( P_{ \cP }, M_{\cP} ; U | \cP, T)}{\ell} \leq \frac{H(U)}{\ell} \leq \frac{s}{\ell}.
    \end{align*}
    This completes the proof of the lemma.
\end{proof}

\subsubsection{Good Advantage} \label{sec:advantage}

We also have the following simple observation on the advantage of \pref{prot:AliceBob_adaptive} over a random guessing.

\begin{lemma}[Advantage]
Fix any $\vec{M}$. For any $i \in [m]$, the probability of \pref{prot:AliceBob_adaptive} outputting $\chi_{M_i} (X)$ correctly over $X$ is at least $\frac{1 +  2^{-t}}{2}$
\end{lemma}
\begin{proof}

    If we fix $\vec{M}$, first observe that $X$ is independently at random. As $\vec{M}$ completely determines Alice's message, $P_i$, the only remaining part of the protocol is $B_i$ and $Z_i$. There are two possible scenarios for $B_i$. If $U$ agrees with $P_i$, or not, that is if $B_i = 1$ or not. If $B_i = 1$, this implies that the last bit in $P_i$ is indeed $\chi_{M_i} ( X )$ due to the correctness of the original $t+1$-round Communication Process~\pref{prot:multiphase}. If $B_i = 0$, the process takes a random guess. Therefore, the probability of being correct can be written exactly as 
    \begin{equation*}
        \Pr_{X} [ B_i = 1 ] + \Pr_{X} [ B_i = 0] \cdot \frac{1}{2} 
         = \Pr_{X} [ B_i = 1 ] + ( 1 - \Pr_{X} [ B_i = 1 ] ) \cdot \frac{1}{2} = \frac{1 +  \Pr_{X} [ B_i = 1 ]}{2}
    \end{equation*}
    Therefore it suffices to bound $\Pr_{X} [ B_i = 1 ]$. Observe that we choose $P_i$ as the maximum likelihood path. 
    Therefore $\Pr_{X} [ P_i ] \geq \frac{1}{2^t}$, as there are at most $2^t$ many possible $P_i$. 
    Denote $P'_i$ as the true decision tree path given by fixed $\vec{M}$ and $T$. Then $\Pr_{X} [ B_i = 1 ] = \Pr_{X} [ P_i = P'_i ] = \Pr_{X} [ P_i ]$, that is the probability of path $P_i$ over possible $T$'s. Then, we get 
    \begin{equation*}
        \Pr_{X} [ B_i = 1 ] \geq \frac{1}{2^t}
    \end{equation*}
    Therefore, the probability of outputting $\chi_{M_i} (X)$ correctly is at least
    \begin{equation*}
        \frac{1 +  \Pr_{X} [ B_i = 1 ]}{2} \geq \frac{1 + 2^{-t}}{2}
    \end{equation*}
\end{proof}

This immediately implies the following corollary, which we will use towards our final contradiction.
\begin{corollary} \label{cor:advantage}
For any setting of $Z_{Alice} = z_{Alice} , M_{\rho_j}$, and $E = 1$ (which are all completely determined by $\vec{M}, \cP, J$)
    \begin{equation*}
        \Pr[ Z_{\rho_j} = \chi_{M_{\rho_j}} (X) | M_{\rho_j}, Z_{Alice} = z_{Alice}, E = 1 ] \geq \frac{1 + 2^{-t}}{2} 
    \end{equation*}
\end{corollary}

\subsection{Combinatorial Lemma} \label{sec:combinatorial}

The main corresponding combinatorial lemma to the \cite{ko_adaptive_2020, dvorak_lower_2020} is the following. This shows that no random process that achieves even the slightest advantage over random guessing with low correlation and large average min-entropy can exist. 

\begin{lemma} \label{lem:comb}
    No setting of a random process $Z = z$ and $C$, which contains $z_{out}$ can simultaneously satisfy all three of the following inequalities for any $\gamma \geq 3$
    \begin{align}
    & \E_{C | Z = z } \left[ \abs{  \E_{M_i, X |_{C = c, Z = z}} \left[ \chi_{M_i} (X) \cdot z_{out} | C = c , Z = z \right] } \right] \geq 2^{- \gamma + 2} \label{eq:condition1} \\
    & \tilde{H}_\infty ( M_i | C, Z = z ) + \tilde{H}_\infty ( X | C, Z = z ) \geq n + 2 \gamma  \label{eq:condition2} \\
    & I ( M_i; X | C, Z = z ) \leq 2^{-2 \gamma} \label{eq:condition3}
    \end{align}
\end{lemma}


\begin{proof}
    For the sake of contradiction, suppose such $z$ and $C$ exists. First as $z_{out}$ is $\pm 1$, we can write 
    \begin{equation}
        \abs{  \E_{M_i, X |_{C = c, Z = z}} \left[ \chi_{M_i} (X) \cdot z_{out} | C = c , Z = z \right] } =  \abs{ \E_{ M_i, X |_{C = c, Z = z}} \left[ \chi_{M_i} (X) | C = c , Z = z \right] }
    \end{equation}
    Then we can use the $\ell_1$ bound to have
    \begin{align}
        \abs{ \E_{ M_i, X |_{C = c, Z = z}} \left[ \chi_{M_i} (X) | C = c , Z = z \right] } & \leq \abs{ M_i |_{C = c, Z = z} \cdot H \cdot X |_{C = c, Z = z} } \label{eq:invariant_hadamard} \\
        & +  \norm{ M_i |_{C = c, Z = z} \times X |_{C = c, Z = z} - (M_i,X) |_{C = c, Z = z} }_1  \label{eq:invariant_rectangle} 
    \end{align}

    We bound the expectation of \pref{eq:invariant_rectangle}. Our KL-divergence term is then equal to the mutual information between $M_i$ and $X$ conditioned on $Z = z$. Namely, using the chain rule for the KL divergence,
    \begin{align*}
        D ( M_i, X |_{Z = z}  || M_i |_{Z = z} \times X |_{ Z = z } ) & = \underbrace{ D( X |_{Z = z} || X |_{Z = z} )}_{ = 0} + \E_{x \sim X |_{Z = z}} \left[ D ( M_i|_{X = x, Z = z}  || M_i |_{Z = z} )  \right] \\
        & = I ( M_i; X | Z = z ) 
    \end{align*} 
    Then, due to Pinsker's inequality (\pref{fact:pinsker}), we have
    \begin{align*}
        & \norm{ M_i |_{C = c, Z = z} \times X |_{C = c, Z = z} - (M_i,X) |_{C = c, Z = z} }_1 \leq 2 \sqrt{ I ( M_i ; X | C = c, Z = z ) }
    \end{align*}
    Then taking expectation over $C$ and applying Jensen's inequality,
    \begin{equation} \label{eq:invariant_rectangle_expectation}  
        \E_{C |_{Z = z}} \left[ \norm{ M_i |_{C = c, Z = z} \times X |_{C = c, Z = z} - (M_i,X) |_{C = c, Z = z} }_1 \right] \leq 2 \sqrt{ I ( M_i ; X | C, Z = z ) } \leq 2^{-\gamma + 1} 
    \end{equation}
    where the last bound holds from \pref{eq:condition3}.

    Next, we bound \pref{eq:invariant_hadamard}. Due to \pref{cl:hadamard} and Cauchy-Schwarz Inequality,
    \begin{align*}
        & \E_{C |_{Z = z}} \left[ \abs{ M_i |_{C = c, Z = z} \cdot H \cdot X |_{C = c, Z = z} } \right] \leq \E_{C |_{Z = z}} \left[ 2^{n/2} \cdot \norm{ M_i |_{C = c, Z = z} }_2 \cdot \norm{X |_{C = c, Z = z} }_2 \right] \\
        & \leq 2^{n/2} \cdot \sqrt { \E_{C |_{Z = z}} \left[ \norm{ M_i |_{C = c, Z = z} }_2^2 \right] \cdot \E_{C |_{Z = z}} \left[ \norm{ X |_{C = c, Z = z} }_2^2 \right] }
    \end{align*}
    \pref{prop:l-2-min} implies
    \begin{align*}
        & \E_{C |_{Z = z}} \left[ \norm{ M_i |_{C = c, Z = z} }_2^2 \right] \leq 2^{ - \tilde{H}_\infty ( M_i | C, Z =z ) } \\
        & \E_{C |_{Z = z}} \left[ \norm{ X |_{C = c, Z = z} }_2^2 \right] \leq 2^{ - \tilde{H}_\infty ( X | C, Z =z ) }
    \end{align*}
    which would in turn imply 
    \begin{align*}
        &  \sqrt { \E_{C |_{Z = z}} \left[ \norm{ M_i |_{C = c, Z = z} }_2^2 \right] \cdot \E_{C |_{Z = z}} \left[ \norm{ X |_{C = c, Z = z} }_2^2 \right] } \leq 2^{ - ( \tilde{H}_\infty ( M_i | C, Z =z ) + \tilde{H}_\infty ( X | C, Z =z ) ) / 2 } \\
        & \leq 2^{- \frac{n + 2 \gamma}{2}} = 2^{-n/2} \cdot 2^{ -\gamma}
    \end{align*}
    which then implies \pref{eq:invariant_hadamard} is upper bounded by
    \begin{equation}
        \pref{eq:invariant_hadamard} \leq 2^{n/2} \cdot 2^{-n/2} \cdot 2^{-\gamma} = 2^{- \gamma}
    \end{equation}
    
    Therefore, we get
    \begin{align}
    \E_{C | Z = z } \left[ \abs{  \E_{M_i, X |_{C = c, Z = z}} \left[ \chi_{M_i} ( X ) \cdot z_{out} | C = c , Z = z \right] } \right] \leq \frac{3}{2^{\gamma}} 
    \end{align}
    which contradicts \pref{eq:condition1}.
\end{proof}

\subsection{Combining the lemmas} 
Finally, we combine the lemmas from the previous two sections to prove the main theorem via contradiction. Recall the statement of our main theorem.
\begin{reptheorem}{thm:main} [Main]
   Let $t,s,m \geq 100$ be parameters such that $m = |\cQ| \geq \omega ( s \cdot 2^{3 t} ), \log |\cS| \geq 40 \cdot  t \log s$. Consider any data structure which answers $m$ many linear functions from a collection $\cS \subseteq \Field_2^n$ say $\cQ$. There must exist some $\cQ$ such that there is no data structure for $\cQ$ using $s$-space, $t$ probes per query. 
\end{reptheorem}

\begin{proof}
    We will prove via contradiction. Suppose otherwise. Then
    we know that \pref{prot:AliceBob_adaptive} must exist as well with the provided definitions of $Z_{Alice}, E$, and $Z$. Our goal is to show that there exists a setting of $Z_{Alice} = z_{Alice}, E = 1$ and $Z_{\rho_j} B_{\rho_j}$ that would violate \pref{lem:comb}, thus leading to a contradiction.
    
    First, we will fix $Z_{Alice} = z_{Alice}, E = 1$ which satisfies large average Min-Entropy and low correlation simultaneously. That is, we will fix $Z_{Alice} = z_{Alice}, E = 1$ such that
    \begin{align}
        & I ( M_{\cP_J} ; X | Z_{Alice} = z_{Alice} , E = 1, Z_{\rho_j} B_{\rho_j} ) \leq \frac{6 s}{\ell} \label{eq:9} \\
        & \tilde{H}_\infty ( M_{\rho_j} | Z_{\rho_j} B_{\rho_j}, E = 1, Z_{Alice} = z_{Alice} )  + \tilde{H}_\infty ( X | Z_{\rho_j} B_{\rho_j}, E = 1, Z_{Alice} = z_{Alice} ) \nonumber  \\
        & \geq n + \log |\cS| - 300 t \log s \label{eq:10}
    \end{align}
    
    Over the random choice of $Z_{Alice} = z_{Alice}$, \pref{lem:S} and \pref{lem:l2S} implies that 
    \begin{equation*}
        \tilde{H}_\infty ( M_{\rho_j} | Z_{\rho_j} B_{\rho_j}, E = 1, Z_{Alice} = z_{Alice} )  + \tilde{H}_\infty ( X | Z_{\rho_j} B_{\rho_j}, E = 1, Z_{Alice} = z_{Alice} ) \geq n + \log |\cS| - 30 t \log s
    \end{equation*}
    \pref{lem:corr} implies that over the random choice of $Z_{Alice} = z_{Alice}$, 
    \begin{equation*}
        I ( M_{\cP_J} ; X | Z_{Alice} = z_{Alice} , E = 1,  Z_{\rho_j} B_{\rho_j} ) \leq \frac{2 s}{\ell} 
    \end{equation*}
    Due to Markov argument, there exists $Z_{Alice} = z_{Alice}$ which simultaneously satisfy both \pref{eq:9} and \pref{eq:10}. 

    On the other hand, \pref{cor:advantage} implies that for any setting of $Z_{Alice} = z_{Alice}, E = 1$ and $S_{\rho_j}$ 
    \begin{equation*}
        \Pr[ Z_{\rho_j} = \chi_{M_{\rho_j}} (X) | M_{\rho_j}, Z_{Alice} = z_{Alice}, E = 1 ] \geq \frac{1 + 2^{-t}}{2} 
    \end{equation*}
    or equivalently
    \begin{equation} \label{eq:11}
        \E_{ M_{\rho_j}, X |_{Z_{\rho_j} B_{\rho_j}, Z_{Alice} = z_{Alice}, E = 1}} \left[ \abs{ \chi_{M_{\rho_j}} (X) \cdot Z_{\rho_j} } \right] \geq 2^{-t}
    \end{equation}

    Then consider \pref{eq:9}, \pref{eq:10} and \pref{eq:11}. To obtain a desired contradiction with \pref{lem:comb}, setting $\gamma = t + 2$ in the lemma, it suffices to have
    \begin{align*}
        & n + \log |\cS| - 30 t \log s \geq n + 2 ( t + 2 ) \\
        & \frac{ 6 s }{\ell} \leq 2^{- 2 ( t + 2)}
    \end{align*}
    which is implied by
    \begin{align*}
        & 40 t \log s  \leq \log |\cS| \\
        & 6 s 2^{ 2 ( t + 2 )} \leq s \cdot 2^{3t} \leq \ell = |\cQ| / 100
    \end{align*}
    These conditions are implied by the choice of parameters given in the statement of the theorem. This would contradict \pref{lem:comb}, completing the proof of the theorem.
\end{proof}
\section{Wire Lower Bound for Circuits with Arbitrary Gates} \label{sec:viola}

In this section, we show that a random linear operator satisfies one of the conditions laid out by \cite{viola_lower_2018} to obtain a breakthrough for lower bounds in circuits with arbitrary gates. This shows that a random linear operator does beat the state-of-the-art lower bounds given in \cite{cherukhin_lower_2008,hirsch_lower_2008, gal_tight_2013}.

\subsection{Circuit with Arbitrary Gates}

In this section, we formally define a circuit with arbitrary gates. (See Chapter 13 of \cite{jukna_boolean_2012} and \cite{drucker_limitations_2012} for further references) 
We would like to compute a Boolean operator $f : \{ 0,1 \}^n \to \{ 0 , 1 \}^m$ using a circuit where a gate can compute any function with unbounded fan-in. Note that this is the strongest possible model of circuit as computations are given for free. Therefore, its lower bound should apply to all possible models of circuit. Furthermore, it is meaningless to count the number of gates, as any $f$ can be computed with $m$ gates. What we measure instead is information transfer, quantified by the number of wires. Then a trivial upper-bound is $n \cdot m$, while a trivial lower bound is $\max \{ n, m \}$. Assuming $m > n$, a non-trivial lower bound on the number of wires would be of the form $\omega( m )$. 

In this unrealistically powerful model, we want to study the trade-off between the number of wires required for $f$ as a function of the depth of the circuit $d$, output size $m$ and input size $n$. Here we consider the setting where $m = O(n)$.

\subsubsection{Previous Results}
In order to describe the previous results, we need the following definition.
\begin{definition}
    $\lambda_1 ( n ) = \lfloor \sqrt{n} \rfloor$, $\lambda_2 (n) = \lceil \log n \rceil$. For $d \geq 3$, $\lambda_d (n) := \lambda_{d-2}^* (n)$ where $^*$ denotes the number of times the function must be applied to $n$ to reach a value $\leq 1$.
\end{definition}

Due to a simple counting argument, most arbitrary operators require $\tilde{\Omega} ( n^2 )$-wires \cite{jukna_circuits_2010}. However, if we turn to finding an explicit (or even semi-explicit) operator with $\omega( n )$-wires, the best explicit bound known (an improvement over \cite{raz_lower_2003}) is due to Cherukhin’s Bound $\Omega_d ( n \cdot \lambda_{d-1} (n) )$ \cite{hirsch_lower_2008, cherukhin_lower_2008, drucker_limitations_2012}. However, the operator in consideration is not (and cannot be) a linear operator, though explicit. On the other hand, if we turn to finding some hard linear operator, \cite{gal_tight_2013} shows that computing any ``good" linear error correcting code\footnote{most linear operators are ``good" linear error correcting codes due to Gilbert-Varshamov bound holding for random linear operators \cite{guruswami_essential_2022}. And there are explicit constructions of ``good" linear error correcting codes.} requires $\Omega( n \cdot \lambda_{ 2 \cdot \lfloor d/2 \rfloor} (n) )$-many wires. 

\begin{table}[!h]
\centering
\begin{tabular}{|l|l|l|l|}
\cline{1-4}
Depth & \cite{hirsch_lower_2008, cherukhin_lower_2008} & \cite{gal_tight_2013} & Our Result \\ \cline{1-4}
$d=2$   & $\Omega (n \cdot \sqrt{n} )$  & $\Omega (n \cdot ( \log (n) / \log \log (n) )^2 )$ & $\Omega (n \cdot \log^{1/2} (n) )$ \\ 
$d=3$   & $\Omega (n \cdot \log (n) )$  & $\Omega (n \cdot \log \log (n) )$ & $\Omega (n \cdot \log^{1/3} (n) )$ \\
$d=4$   & $\Omega (n \cdot \log\log (n) )$ & $\Omega (n \cdot \log^* (n) )$ & $\Omega (n \cdot \log^{1/4} (n) )$ \\ 
$d$     & $\Omega_d (n \cdot \lambda_{d-1} (n) )$ & $\Omega (n \cdot \lambda_{ 2 \cdot \lfloor d/2 \rfloor } (n) )$ & $\Omega (n \cdot \log^{1/d} (n) )$ \\ \cline{1-4}
\end{tabular}
\caption{Wire lower bounds for Circuits with Arbitrary Gates \label{tab:result}}
\end{table}

The primary technical reason for the rapid decay of wire lower bounds with increasing depth is the reliance on the so-called {\em superconcentrator} technique \cite{valiant_non-linear_1975, pippenger_superconcentrators_1977, dolev_superconcentrators_1983, pudlak_communication_1994, pudlak_combinatorial-algebraic_1994, radhakrishnan_bounds_2000, gal_tight_2013} or more refined {\em Strong Multiscale Entropy} (SME) approach \cite{raz_lower_2003, hirsch_lower_2008, cherukhin_lower_2008, jukna_boolean_2012}. These are the only previously known methods--aside from naive counting--for establishing lower bounds on circuits with arbitrary gates. Moreover, the limitations of these techniques are well-documented. For instance, \cite{dolev_superconcentrators_1983, pudlak_communication_1994, radhakrishnan_bounds_2000, gal_tight_2013} exhibited superconcentrators with small circuits. \cite{drucker_limitations_2012} demonstrated that neither the SME approach nor a generalization of \cite{gal_tight_2013} can yield further improvements.

\subsubsection{Our Result}
We show a wire lower bound which only suffers a polynomial decay per depth $d$ compared to the inverse Ackermann function type bound in previous works at the cost of using a random linear operator. In particular, we show the following theorem.
\begin{theorem} \label{thm:circuitLB}
For most linear operators $M : \{0,1\}^n \to \{0,1\}^{O(n)}$, when computed by a circuit of depth $d$ requires
\begin{equation*}
    w \geq \Omega( n \cdot \log^{1/d} (n) )
\end{equation*}   
\end{theorem}
For depth-2 circuit, our bound is weaker than that of \cite{gal_tight_2013}. But we get a super-exponential improvement for $d \geq 3$.
We also beat Cherukhin’s Bound \cite{hirsch_lower_2008, cherukhin_lower_2008} for $d \geq 4$. This also implies a superlinear wire lower bound as long as $d = o ( \log \log n)$. The result is summarized in \pref{tab:result}.

We remark that such lower bound is only possible with an approach that is drastically different from superconcentrator or SME due to known limitations of these techniques \cite{dolev_superconcentrators_1983, pudlak_communication_1994, radhakrishnan_bounds_2000, drucker_limitations_2012}. 

\paragraph{Separation between Representing a Linear Operator} This also gives an exponentially stronger separation between ``representing" a linear operator \cite{jukna_representing_2010, drucker_limitations_2012, jukna_boolean_2012} and ``computing" a linear operator under circuits with arbitrary gates. A circuit $C$ {\em represents} a linear operator $f$ if there exists some basis $B \subset \Field_2^n$ such that for every $b \in B$, $f(b) = C(b)$. Note that if a circuit is a {\bf linear} circuit, representing and computing are equivalent tasks. But for {\bf non-linear} circuits, they are not necessarily equivalent. Drucker~\cite{drucker_limitations_2012} showed that most linear operators can be represented by a circuit with $O(n)$ wires in depth-$3$, while \cite{gal_tight_2013} shows $\Omega( n \log \log n )$-wire lower bound for computing a linear operator, giving a separation for two tasks for non-linear circuits. Our result implies $\Omega ( n \log^{1/3} n )$ lower bound.

\subsection{Proof of \pref{thm:circuitLB}}

We use the following theorem to translate the cell-probe lower bound to circuit wire lower bound.
\begin{theorem}[Theorem 2.3 of \cite{viola_lower_2018}] \label{thm:viola}
Suppose the operator $f : \{0,1\}^n \to \{0,1\}^m$ has a circuit of depth $d$ with $w$ wires, consisting
of unbounded fan-in, arbitrary gates. Then $f$ has a data structure with space $s = n+r$ and time $(w/r)^d$ for any $r$.
\end{theorem}

We show the following lemma which follows from modifying the proof of \pref{thm:main}. Note that this is stronger than what is necessary as we show a lower bound against small constant advantage over random guessing.
\begin{lemma} \label{lem:tight_log_bound}
    For random linear operator $M : \{0,1\}^n \to \{0,1\}^m$ with $m = 10^9 \cdot n$, 
    any cell-probe data structure which correctly outputs $f_i ( x )$ with probability at least $2/3$ for all $i \in [m]$ using $s = 1.01 n$-space must have $t \geq \Omega( \log n )$. 
\end{lemma}
Note that \pref{lem:tight_log_bound} and \pref{thm:viola} directly imply \pref{thm:circuitLB}.

\begin{proofsketch}
    Here we sketch and highlight the required modification to the main proof. Suppose there exists a cell-probe data structure for most random linear operator $f : \{0,1\}^n \to \{0,1\}^m$ using $s = 1.01 n$-space with $t \leq 0.1 \log n$. 

    We replace \pref{prot:AliceBob_adaptive}. We redefine $P_i$ to represent the entire decision tree rather than a single path within it. Instead of sending a single path in the decision tree, as $t \leq 0.1 \log n$, we can afford to send the whole decision tree using $2^t \cdot \log s = o ( n^{0.2} )$-bits. Then Bob can compute the whole outcome of the decision tree using $U$. Furthermore, the correctness of the cell-probe data structure guarantees Alice and Bob correctly output the answer with probability $> 2/3$. Therefore, the random variable $B_i$ is unnecessary and we can have $Z_i$ as the output.
    
    With this change, we leave the reader to verify that we can replace \pref{eq:9} and \pref{eq:10} in \pref{thm:main} as 
    \begin{align}
        & I ( M_{\cP_J} ; T | Z_{Alice} = z_{Alice} , E = 1, Z_{\rho_j}  ) \leq \frac{6 s}{\ell} \leq \frac{1}{10^5}  \\
        & \tilde{H}_\infty ( M_{\rho_j} | Z_{\rho_j}, E = 1, Z_{Alice} = z_{Alice} )  + \tilde{H}_\infty ( X | Z_{\rho_j}, E = 1, Z_{Alice} = z_{Alice} ) \nonumber  \\
        & \geq 2 n - n^{0.2} 
    \end{align}
    while $\abs{ Z_{\rho_j} \cdot \chi_{M_{\rho_j}} ( X ) }$ is $\geq 1/3$ in expectation. This then violates \pref{lem:comb} by taking $\gamma$ as the appropriate constant. 
\end{proofsketch}

\begin{remark}
    The simple modification can also be adapted to a carefully manipulated version of the original argument in \cite{ko_adaptive_2020, dvorak_lower_2020}, thereby giving the same $t \geq \Omega( \log n )$ bound when the data structure is guaranteed to output the correct answer all the time. However, as noted from the sketch of the proof, the new argument works against circuits obtaining some $\Omega(1)$ advantage over random guessing. That is our proof works against the circuit $C$'s such that for every $i \in [m]$
    \begin{equation*}
        \Pr_{x} \left[ C_i ( x ) = f_i ( x ) \right] \geq \frac{1 + \Omega(1)}{2}
    \end{equation*}
    where $C_i$ denotes the $i$-th output of the circuit $C$. The argument from \cite{ko_adaptive_2020, dvorak_lower_2020} cannot be used to handle such small advantage regime.
\end{remark}
\section{Acknowledgment}
We thank Sasha Golovnev for his feedback in the earlier draft of the paper.
A part of this work was done while attending DIMACS Workshop on Lower Bounds and Frontiers in Data Structures 2022.

\newpage 
\bibliographystyle{alpha}
\bibliography{references.bib}

\appendix
\section{Omitted Proof}

We restate the lemma.
\begin{replemma}{lem:S}
    If $|\cP| = \ell = \frac{|\cQ|}{100} = \frac{m}{100}$ then
    \begin{equation*}
        I ( M_{\cP_J} ; Z_{Alice} ) = E_{z_{Alice}} \left[ D( M_{\cP_J} |_{Z_{Alice} = z_{Alice} } || M_{\cP_J} ) \right]  \leq 3 \cdot t \log s 
    \end{equation*}
\end{replemma}
\begin{proof}
    We plug in the definition of $Z_{Alice}$.
    \begin{align*}
        & I ( M_{\cP_J} ; Z_{Alice} ) = I ( M_{\cP_J} ; \cP, J, P_{ \cP_{<J} }, M_{ \cP_{<J} }, P_{ \cP_{J} }, ) \\
        & = I ( M_{\cP_J} ; \cP, J, P_{ \cP_{<J} }, M_{ \cP_{<J} } ) + \underbrace{ I (  M_{\cP_J} ; P_{ \cP_{J} }  | \cP, J, P_{ \cP_{<J} }, M_{ \cP_{<J} } )}_{ \leq t \log s + t + 1} \\
        & \leq I ( M_{\cP_J} ; \cP, J, P_{ \cP_{<J} }, M_{ \cP_{<J} } ) + t \log s + t + 1
    \end{align*}
    Now we upper bound $I ( M_{\cP_J} ; \cP, J, P_{ \cP_{<J} }, M_{ \cP_{<J} } )$ term. We first know that due to $J$ and $\cP$ being chosen independently at random
    \begin{equation*}
        I ( M_{\cP_J} ;  \cP J, P_{ \cP_{<J} }, M_{ \cP_{<J} } ) = I ( M_{\cP_J} ; \cP, P_{ \cP_{<J} }, M_{ \cP_{<J} }  | \cP, J )
    \end{equation*}
    Now consider a fixed $J = j$. 
    \begin{align*}
        & I ( M_{\cP_J} ; \cP, P_{ \cP_{< j} },  M_{ \cP_{<J} } | \cP, J = j )  = I ( M_{\cP_J} ; P_{ \cP_{< j} },  M_{ \cP_{<J} } | \cP ) = I ( M_{\cP_J} ; P_{ \cP_{< j} } | M_{ \cP_{<J} }, \cP) 
    \end{align*}
    Now consider fixed $\cP_{ < j } = \rho_{<j}$. Then the above term becomes
    \begin{align*}
        & I ( M_{\cP_J} ;  P_{\rho_{<j}} | M_{ \rho_{<j} } , \cP_{<j} = \rho_{<j} , \cP_j ) = \frac{1}{m - (j-1)} \sum_{i \notin \rho_{<j} } I ( M_i ;  P_{\rho_{<j}} | M_{ \rho_{<j} } ) \\
        & \leq \frac{1}{m - (j-1) } I ( S_{[m] - \rho_{<j}} ; P_{\rho_{<j}} | M_{ \rho_{<j} } ) \leq \frac{j}{m - \ell}  ( t \log s + t + 1 ) \leq 0.1 t \log s 
    \end{align*}
    due to \pref{fact:chainrule1} and the last inequality holds due to our choice of $\ell$

    As the inequality holds for any fixed $j$ and $\cP_{ < j } = \rho_{<j}$, 
    \begin{equation*}
        I ( M_{\cP_J} ; Z_{Alice} ) \leq t \log s + t + 1 + 0.1 t \log s \leq 3 \cdot t \log s
    \end{equation*}
    completing the proof of our lemma.

\end{proof}

\end{document}